\documentclass[alpha]{article}
\usepackage{amssymb}
\usepackage{amsthm}	% extended theorem environments ("proof")
\usepackage{graphicx}
\usepackage{subfig}
\usepackage[english]{babel}

\usepackage{color}
\newcommand{\SAT}{{\tt SAT}}
\newcommand{\cL}{{\cal L}}

\newcommand{\cLM}{{\cal L}^M}
\newcommand{\cWm}{{\cal W}^m}
\newcommand{\cW}{{\cal W}}

\newcommand{\cX}{{\cal X}}
\newcommand{\cF}{{\cal F}}

\newcommand{\OTR}[2]{$({#1}, {#2})$-{\tt Trade}}
\newcommand{\TR}[3]{$({#1}, {#2}, {#3})$-trade}
\newcommand{\SetS}{{\tt SetSplitting}}
\newcommand{\NP}{{\tt NP}}
\newcommand{\coNP}{{\tt coNP}}
\newcommand{\OMIT}[1]{}

\newtheorem{theorem}{Theorem}[section]
\newtheorem{definition}[theorem]{Definition}
\newtheorem{proposition}[theorem]{Proposition}

\newtheorem{example}[theorem]{Example}
\newtheorem{conjecture}[theorem]{Conjecture}
\begin{document}

%\title{On the Complexity of Trading\tnoteref{t0}}
\title{On the Complexity of Exchanging\footnote{This work is partially supported by grant 2014SGR1034 (ALBCOM) of ``Generalitat de Catalunya''.}}

\author{X.~Molinero~\thanks{Department of Applied Mathematics III. Universitat Polit\`ecnica de Catalunya, Manresa, Spain. E-mail: {\tt xavier.molinero@upc.edu}. Partially funded by grant MTM2012-34426/FEDER of the ``Spanish Economy and Competitiveness Ministry''.}
	\and
	M.~Olsen~\thanks{AU Herning, Aarhus University, Denmark. E-mail: {\tt martino@auhe.au.dk}}
	\and
	M.~Serna~\thanks{Department of Computer Science, Universitat Polit\`ecnica de Catalunya, Barcelona, Spain. E-mail: {\tt mjserna@cs.upc.edu}. Partially funded  by grant TIN2013--46181-C2-1-R (COMMAS)  of  the ``Ministerio de Econom\'{\i}a y Competitividad''.}}

\maketitle

\begin{abstract}
We analyze the computational complexity of the problem of deciding whether, for a given simple game, there exists the possibility of rearranging the participants in a set of $j$ given losing coalitions into a set of $j$ winning coalitions. We also look at the problem of turning winning coalitions into losing coalitions. We analyze the problem when the simple game is represented by a list of wining, losing, minimal winning or maximal loosing coalitions.
\newline


\end{abstract}

\section{Introduction}

Simple games cover voting systems in which a single alternative, such as a bill or an amendment,
is pitted against the status quo. In these systems, each voter responds with a vote of yea and nay.
Democratic societies and international organizations use a wide variety of complex rules to reach
decisions. Examples, where it is not always easy to understand the consequences of the way
voting is done, include the Electoral College to elect the President of the United States,
the United Nations Security Council, the governance structure of the World Bank, the International Monetary Fund, the
European Union Council of Ministers, the national governments of many countries, the councils in several counties, and the system to elect the major in cities
or villages of many countries. Another source of examples comes from economic
enterprises whose owners are shareholders of the society and divide profits or losses
proportionally to the numbers of stocks they posses, but make decisions by voting
according to a pre-defined rule (i.e., an absolute majority rule or a qualified
majority rule).
See~\cite{Tay95,TaZw99} for a thorough presentation of theses and other examples.
Such systems have been analyzed as \emph{simple games}.

%There are a lot of results about simple games for tradeness,
%trade robustness and weightness.
%Here we present some new complexity results related with them.
%Most of the notation and definitions follows from~\cite{CaFr96,TaZw99}.

%A \emph{simple game} $G$ is a pair $(N,\cW)$ in
%which $N = \{1,2,\dots,n\}$ and $\mathcal W$ is a collection of
%subsets of $N$ that satisfies three items: \emph{(1)} $N \in \mathcal W$,
%\emph{(2)} $\emptyset \notin \mathcal W$ and \emph{(3)} the
%\emph{monotonicity} property: $S \in \mathcal W$ and $S \subseteq T
%\subseteq N$ implies $T \in \cW$.
\begin{definition}A \emph{simple game} $\Gamma$ is a pair $(N,\cW)$ in
which $N = \{1,2,\dots,n\}$ and $\mathcal W$ is a collection of
subsets of $N$ that satisfies: \emph{(1)} $N \in \mathcal W$,
\emph{(2)} $\emptyset \notin \mathcal W$ and \emph{(3)} the
\emph{monotonicity} property: $S \in \mathcal W$ and $S \subseteq T
\subseteq N$ implies $T \in \mathcal W$.
\end{definition}

The subsets of $N$ are called \emph{coalitions},
the coalitions in $\cW$ are called \emph{winning coalitions},
and the coalitions that are not winning are called \emph{losing coalitions}
(noted by $\cL$).
Moreover, we say that a coalition is minimal winning (maximal losing) if it is a winning
(losing) coalition all of whose proper subsets (supersets)
are losing (winning). Because of monotonicity, any simple game is completely
determined by its set of minimal winning (maximal losing) coalitions
denoted by $\cWm$ ($\cLM$).
Note that a description of a simple game $\Gamma$ can be given by $(N,\cX)$, where
$\cX$ is $\cW$, $\cL$, $\cWm$ or $\cLM$, see~\cite{TaZw99}.
We focus on the process of exchanging or trading where a motivating example is the following:

\begin{example}
Consider two English football clubs that are in
trouble and in danger of leaving Premier League.
Maybe the two clubs could trade with each other
and exchange players so they both could avoid relegation.
We consider the complexity of figuring out if such an exchange is
possible for various ways of knowing what it takes to form a strong
team that is able to stay in Premier League.
This can be viewed as a simple game where a winning coalition
corresponds to a strong team of players.
\end{example}

The considered property is the so called $j$-trade property for  simple games.
Loosely speaking, a simple game is $j$-trade if it is possible to rearrange
the players in a set of $j$ winning (losing) coalitions into a set of $j$ losing (winning)
coalitions, in such a way that the total number of occurrences of each player is the same in both sets.
 Thus, it is possible to go from one set to the other via participant trades.
This notion was
introduced by Taylor and Zwicker~\cite{TaZw99}
in order to obtain a characterization of the weighted games, a subfamily of simple games.
Recall that any simple game can be expressed as the intersection of weighted simple games.
This leads to the definition of  the \emph{dimension} concept, the minimum number of required
weighted games whose intersection represents the simple game~\cite{DeWo06,FrPu08,Fre04}.
 Due to this fact, the problem of deciding whether a simple game is weighted has been of interest in  several contexts.
With respect to tradeness, it is  known that a simple game is weighted if and only if it is not $j$-trade
for any non-negative integer $j$~\cite{TaZw99}.
Freixas \emph{et al.}~\cite{FMOS11} studied the computational
complexity of deciding  whether a simple game is weighted among other decision problems for simple games.
In particular,  they showed that
deciding whether a simple  game is weighted is polynomial time solvable
when the game is given by an explicit listing of one of the families $\cW$, $\cL$, $\cWm$, $\cLM$.
On the other hand, the $j$-trade concept was also
redefined as $j$-invariant-trade of simple games~\cite{FrMo09}
and extended as $(j,k)$-simple games~\cite{FrZw03}.
%This concept is applied to a lot of topics: political, sociology, economics, etc.
%The goal of this paper is to study the so-called $j$-trade simple games.

Here we provide a definition of $j$-trade that uses a formalism that differ from the
classic one for $j$-trade \emph{robustness} applied to a simple game
(see~\cite{CaFr96,TaZw99,FrMo09}) in order to ease the proofs of our new results.

\begin{definition}
\label{def:j-trad-appl}
Given a simple game $\Gamma$, a \emph{$j$-trade application} is a set
of $2j$ coalitions $(S_1,\ldots,S_{2j})$ such that
$\exists I\subseteq\{1,\ldots,2j\}$ that satisfies:
\begin{enumerate}
\item{}$|I|=j$
\item{}$\forall i \in\{1,\ldots,2j\}$, $S_i\in\cW\iff i\in{}I$
%\item{}$\forall p \in{}N$, $|\{i\in{}I\,:\,p\in{}S_i\}| = |\{i\notin{}I\,:\,p\in{}S_i\}|$
\item{}$\forall p \in{}N$, $|\{i\in{}I\,:\,p\in{}S_i\}| = |\{i \in\{1,\ldots,2j\}\setminus{}I\,:\,p\in{}S_i\}|$
\end{enumerate}
\end{definition}

%\begin{definition}
%\label{def:j-trad}
%A simple game $(N, \cW)$ is \emph{$j$-trade} if there is an application $J$
%defined as $J :=<S_1, S_2,\ldots, S_j; T_1, T_2,\ldots, T_j>$, where $S_i\in\cW$
%and $T_i\in\cL$, checking $|{i: p\in S_i}| = |{i: p\in T_i}|, \forall p\in N$.
%\end{definition}

\begin{definition}
\label{def:j-trad}
A simple game $\Gamma$ is \emph{$j$-trade} if it admits a $j$-trade application.
\end{definition}

\begin{example}
The simple game defined by $(N,\cWm)=(\{1,2,3,4\}, \{\{1,3\},\{2,4\}\})$ is
$2$-trade because it admits a $2$-trade application.
For instance,
we can consider the following set of coalitions $(\{1,3\},\{2,4\},\{1,2\},\{3,4\})$
where $\{1,3\},\{2,4\}\in\cW$, but $\{1,2\},\{3,4\}\in\cL$.
\end{example}

\begin{example}
It is easy to generate a simple game that will be $2j$-trade, for an integer $j$.
For instance, we can take the simple game $(N,\cWm)$ where
$N=\{1,\ldots,2j\}$ and $\cWm=\{\{i,i+1\} \,|\, i\in{1,3,5,\ldots,2j-1}\}$.
It is clear that coalitions $L_i=\{i,i+1\}$, for all $i\in\{2,4,6,\ldots,2j-2\}$,
and $L_ {2j}=\{1,2j\}$ are losing.
Thus, the set of $2j$ coalitions
$\cWm\cup\left(\cup_{i=1}^j L_{2i}\right)$ generates a $j$-trade application.
\end{example}

\begin{definition}
\label{def:j-trad-robust}
A simple game $\Gamma$ is \emph{$j$-trade robust} if it is not \emph{$j$-trade}.
\end{definition}

%Observe that some papers consider majority games as $j$-trade robust games,
%see , for instance, \cite{FrMo09} [Freixas, Freixas \& Molinero...].
%\comment{Xavier will find more references, specially EJOR's references.}

Before formally defining the \emph{decision problems} we focus on,
we consider two functions $\alpha$ and $\beta$ associating games with
various types of sets of coalitions.
The allowed types are the following
$\alpha(\Gamma)\in\{\cW,\cL,\cWm,\cLM\}$
and $\beta(\Gamma)\in\{\cW,\cL\}$, respectively.
Moreover, given the  $\beta$ application
we consider the function $\overline{\beta}$ that
provides the \emph{complementary} type with respect
to the function $\beta$.

\[
 \overline{\beta}(\Gamma) =\left\{
\begin{array}{ll}
                \cW, & \textit{if } \beta(\Gamma)=\cL\\
                \cL, & \textit{if } \beta(\Gamma)=\cW\\
               \end{array}
        \right.
\]

Now we can state the definition of the considered computational problems, observe that the value of $\alpha$ provides the type of coalitions used in the representation of the input game while the $\beta$ function indicates the type of the coalitions to be exchanged.

\begin{definition}
The \emph{\TR{\alpha}{\beta}{j}} problem, where $j\in{}\mathbb{N}$, is
\begin{description}
\item{\tt Input:} A simple game $\Gamma$ given by $(N,\alpha(\Gamma))$  and $j$ coalitions $S_1,\dots,S_j\in{}\beta(\Gamma)$.
\item{\tt Question:} Do there exist $S_{j+1},\ldots,S_{2j}\in\overline{\beta}(\Gamma)$ such that
$(S_1, \ldots, S_{2j})$ is a $j$-trade application?
\end{description}
\end{definition}

\begin{definition}
The \emph{\OTR{\alpha}{\beta}} problem is the \TR{\alpha}{\beta}{2} problem.
\end{definition}

In the remaining part of the paper we analyze the computational complexity of the above problems. Table~\ref{tab:sum} summarizes all results about the \emph{\OTR{\alpha}{\beta}} problem.
We present first the results for the \OTR{\alpha}{\beta} problem and then the results for the general case. We finalize with some conclusions and open problems.
\begin{table}[t]
%\begin{table}[h]
\centering
\begin{tabular}{|c|c|}
\hline
 $\alpha(\Gamma)\setminus\beta(\Gamma)$ & $\cW$\\
\hline
$\cW$                                   & {polynomial}\\
$\cW^m$                                 & \NP-complete\\
$\cL$                                   & {polynomial}\\
$\cL^M$                                 & {polynomial}\\
\hline
\end{tabular}

\medskip

\begin{tabular}{|c|c|}
\hline
 $\alpha(\Gamma)\setminus\beta(\Gamma)$ & $\cL$ \\
\hline
$\cW$                                   & {polynomial}\\
$\cW^m$                                 & {polynomial}\\
$\cL$                                   & {polynomial}\\
$\cL^M$                                 & \NP-complete\\
\hline
\end{tabular}
\label{tab:sum}
\caption{Complexity for the \OTR{\alpha}{\beta} problem, where \texttt{polynomial} means polynomially time solvable.}
\label{tab:sum}
\end{table}

\section{The computational complexity of trading two given coalitions}

We present  first the types for which the \OTR{\alpha}{\beta} problems are polynomial time solvable.
%We start by analyzing the  \OTR{\alpha}{\beta} problem.
\begin{proposition}
\label{prop:1-2-3}
The \OTR{\alpha}{\beta} problem is polynomially time solvable when $\alpha(\Gamma)\in\{\cW, \cWm, \cL\}$
and $\beta(\Gamma)=\cL$.
\end{proposition}
\begin{proof} We analyze each case separately.
Let $S_1,S_2$ be two coalitions and assume that both are of type $\beta(\Gamma)=\cL$.
\begin{itemize}
\item[$\bullet$]{} Case $\alpha(\Gamma)=\cW$. Observe that we only need  to check whether there are
two coalitions $S_3, S_4\in\cW$ such that $(S_1, S_2, S_3, S_4)$ is  a $2$-trade application.
This property can be trivially checked in polynomial time by considering all the pairs
of coalitions in $\cW$.
Therefore in polynomial time with respect to the input size.

\item[$\bullet$]{}Case $\alpha(\Gamma)=\cWm$. The algorithm  is the following.
First, we look for the existence of  two coalitions $S_3,\,S_4 \in\cWm$ such that, $\forall p\in N$,
$|\{i\in\{3,4\}: p\in S_i\}| \le |\{i\in\{1,2\}: p\in S_i\}|$. Observe that if such a pair of coalitions exists we can  add the missing players (if any) in such a way that, $\forall p\in N$,
$|\{i\in\{3,4\}: p\in S_i\}| = |\{i\in\{1,2\}: p\in S_i\}|$ and obtain a 2-trade application.

\item[$\bullet$]{}Case $\alpha(\Gamma)=\cL$.  Now we compute $\cW^m$ from $\cL$ using the polynomial time algorithm  shown in~\cite{FMOS11}
and reduce the problem to the previous case.
\end{itemize}%\qed%\hfill \qed
\end{proof}

The same result can be proven when $\beta(\Gamma)=\cW$.

\begin{proposition}
\label{prop:1s-2s-3s}
The \OTR{\alpha}{\beta} problem is polynomially time solvable when $\alpha(\Gamma)\in\{\cL, \cLM, \cW\}$
and $\beta(\Gamma)=\cW$.
\end{proposition}

\begin{proof} Arguments are symmetric to Proposition~\ref{prop:1-2-3}. Let $S_1,S_2$ be two coalitions and assume that both are of type $\beta(\Gamma)=\cW$.

\begin{itemize}
\item[$\bullet$]{} Case $\alpha(\Gamma)=\cL$. Here it is enough to check all pairs of losing coalitions.
The reasoning is symmetric to the first case of Proposition~\ref{prop:1-2-3}.
\item[$\bullet$]{} Case $\alpha(\Gamma)=\cLM$. The algorithm is symmetric to the second case of Proposition~\ref{prop:1-2-3}. We check whether there are two maximal loosing coalitions $S_3$ and $S_4$ so that, $\forall p\in N$,
$|\{i\in\{3,4\}: p\in S_i\}| \ge |\{i\in\{1,2\}: p\in S_i\}|$. If this is the case, by removing the additional players we get a 2-trade application.

\item[$\bullet$]{}If $\alpha(\Gamma)=\cW$, we compute $\cL^M$ from $\cW$ using the polynomial time algorithm given in~\cite{FMOS11}
and use the algorithm for the previous case.
\end{itemize}
\end{proof}

In the following results we isolate the types giving rise to computationally hard cases.

\begin{proposition}
\label{prop:2trade}
The \OTR{\alpha}{\beta} problem is \NP-complete when $\alpha(\Gamma)=\cLM$
and $\beta(\Gamma)=\cL$.
\end{proposition}
\begin{proof}
  The considered \OTR{\alpha}{\beta} problem is easily seen to be a member of \NP. We
  show that it is also \NP-hard  providing a reduction from  the \SAT{} problem.
  Recall that the \SAT{} asks  whether a given boolean
  formula $\phi$ given in conjunctive normal form is satisfiable or not. The
  \SAT{} problem is a famous \NP-complete problem~\cite{GaJo99}. We let
  $X = \{ x_1, \neg x_1, x_2, \neg x_2, \ldots , x_n, \neg x_n \}$
  be  the literals of $\phi$ and let $X_i$ be  the set of  literals in the
  $i$'th clause of $\phi$.  Let $m$ denote the number of clauses
  of $\phi$. Our reduction transforms  $\phi$ into an equivalent
  instance of the considered  \OTR{\alpha}{\beta} problem in polynomial time.

  The set of players of the associated game $\Gamma=\Gamma(\phi)$ contains the literals and
  two extra players $a$ and $b$: $N = X \cup \{a, b\}$. A set of
  players $Y$ can win if and only if at least one of the following two
  conditions are met:
\begin{equation}
\label{eq:satisfy_phi}
a \in Y \wedge \forall i = \{ 1, 2, \ldots, m \}: Y \cap X_i \neq \emptyset
\end{equation}
\begin{equation}
\label{eq:all_xes_covered}
b \in Y \wedge \forall j = \{ 1, 2, \ldots, n \}: x_j \in Y \vee \neg x_j \in Y
\end{equation}
It is not hard to see that this is indeed a simple game since any
superset of a winning set is also winning. We now have to show how to
construct the set of maximal loosing coalitions $L^M$ for this game in
polynomial time. A set of players $S$ is loosing if and only
if~(\ref{eq:satisfy_phi}) and~(\ref{eq:all_xes_covered}) are both
violated. This happens if and only if at least one of the following
four conditions are met:
\begin{equation}
\label{eq:loosing_1}
S \subseteq N \setminus \{a, b\}
\end{equation}
\begin{equation}
\label{eq:loosing_2}
\exists i: S \subseteq N \setminus (X_i \cup \{ b\})
\end{equation}
\begin{equation}
\label{eq:loosing_3}
\exists j: S \subseteq N \setminus \{ a, x_j, \neg x_j\}
\end{equation}
\begin{equation}
\label{eq:loosing_4}
\exists i, j: S \subseteq N \setminus (X_i \cup \{x_j, \neg x_j\}) \enspace
\end{equation}
If we consider all possible combinations of $i$ and $j$ then the sets
on the right hand side of these expressions form a set of loosing sets.
Any loosing set is contained in at least one of those sets.
If we pick the maximal sets of this family -- which
can be done in polynomial time -- we get $\cLM$ for the game
$\Gamma$. The sets $S_1$ and $S_2$ are constructed as follows:
$S_1 = \{a, b\}$ and $S_2 = X$.

Now assume that $\phi$ is a yes-instance to SAT. Let $S_3$ be the set
formed by the player $a$ and all literals corresponding to a
truth-assignment satisfying $\phi$ and let $S_4$ be the set formed by
the remaining literals and the player $b$. It is easy to see that coalitions
$(S_1,S_2,S_3,S_4)$ are a $2$-trade application of $\Gamma$,
where $S_1,S_2\in\cL$ and $S_3,S_4\in\cW$.

%\texttt{Explain more the following paragraph...}

On the other hand, note that such 2-trade application
only exists if one of the winning sets contain $a$ and a set of
literals defining a truth-assignment satisfying $\phi$. Thus,
the instances to the SAT-problem and the considered
\OTR{\alpha}{\beta} problem are equivalent.
\end{proof}

Using a symmetric construction to the previous one we have.
\begin{proposition}
\label{prop:4s}
The \OTR{\alpha}{\beta} problem is \NP-complete when $\alpha(\Gamma)=\cWm$
and $\beta(\Gamma)=\cW$.
\end{proposition}

\OMIT{Note that for all $\phi$ the game $\Gamma$ is $2$-trade? FALSE !}

We conclude this section by isolating a parameter for which one of the hard cases is fixed parameter tractable. We do so by providing a parameterized reduction to the  \SetS{} problem.
\begin{proposition}
The \OTR{\alpha}{\beta} problem, being $\alpha(\Gamma)=\cLM$
and $\beta(\Gamma)=\cL$,  is fixed parameter tractable when considering
the parameter $k = |\cLM|$.
\end{proposition}

\begin{proof}
Recall that according to \cite{GaJo99} the \SetS{} problem is the following \emph{decision problem}:
Given a family $\cF$ of subsets of a finite set $U$ and an integer $k$, decide whether
there exists a partition of $U$ into two subsets $U_1$ and $U_2$ such
that at least $k$ elements of $\cF$ are split by this partition.
%Dehne et al.~\cite{DFR03}
Lokshtanov and Saurabh~\cite{LoSa09} show that the \SetS{} problem is fixed
parameter tractable when the parameter is the integer $k$.

Now we provide a fixed parameter reduction from our case of the \OTR{\alpha}{\beta} problem to the  \SetS{} problem.

Given a simple game $(N, \cL^M)$ and $S_1, S_2 \in \cL$. Let
$Z_i = (S_1 \cup S_2) \setminus L_i$ for $L_i \in \cL^M$ such that
$S_1 \cap S_2 \subseteq L_i$. Let $\mathcal{F}$ be the family of sets consisting
of all the $Z_i$'s. We construct the input to the \SetS{} problem
given by $U=S_1\cup{}S_2$, $\mathcal{F}$ and $k=|\mathcal{F}|$.

Notice that there must be at least one member in $\cF$ since $S_1$
and $S_2$ are loosing coalitions.
Moreover, the number of sets in $\cF$ is always less than or equal to
 $|\cL^M|$.

Now we prove the following claim: Given a simple game $(N, \cL^M)$ and $S_1, S_2 \in \cL$,
there exists $S_3, S_4 \in\cW$ where $(S_1, S_2, S_3, S_4)$ is a
$2$-trade application if and only if $\mathcal{F}$ is a yes-instance to the \SetS{}
problem.

  \begin{bf}If:\end{bf} Now assume that $\cF$ is a
  yes-instance to the \SetS{} problem with the sets $U_1$ and
  $U_2$ splitting the members of $\cF$. We will now prove that
  the sets $S_3=U_1 \cup ( S_1 \cap S_2 )$ and $S_4=U_2
  \cup ( S_1 \cap S_2 )$ are winning coalitions. Now consider an
  arbitrary member $L_i\in\cL^M$. If $S_1 \cap S_2 \subseteq L_i$
  then $U_1 \cap Z_i \neq \emptyset$ implying $S_3 \not
  \subseteq L_i$. If $S_1 \cap S_2 \not \subseteq L_i$ then we also
  have $S_3 \not \subseteq L_i$. This holds for any $L_i \in
  \cL^M$ implying that $S_3$ is winning. The same goes for
  $S_4$.
  Note that $(U_1, U_2)$ is a partition of $S_1 \cup S_2$, so
  every player appears the same number of times in $S_3$ and $S_4$
  as in $S_1$ and $S_2$.

  \begin{bf}Only if:\end{bf} Assume that we have a
 $2$-trade application $(S_1, S_2, S_3, S_4)$ where
 $S_3,S_4\in\cW$ and $S_1,S_2\in\cL$. Let $R = \{L \in
  \cL^M: S_1 \cap S_2 \subseteq L\}$ and $U_i = S_i \setminus (\cap_{L
    \in R} L)$ for $i \in \{3, 4\}$. It is not hard to see that $U_3$
  and $U_4$ are disjoint. What remains is to show that $U_3$ and
  $U_4$ split all the members of $\cF$.

Now consider an arbitrarily chosen set $Z_i$. The coalition
$S_3$ is winning. Therefore  it contains at least one player that
is not a member of $L_i\in\cL^M$. This player is a member of $U_3$
and it is also a member of $Z_i$, so $U_3 \cap Z_i \neq \emptyset$.
Finally, the same argument can be used for $U_4$.

Using the FPT Algorithm for the \SetS{} problem with complexity
$f(k)+ p(n)$, where $k=|\cF|$ and $n=|N|$, for some function
$f$ and a polynomial $p$,
we get that the total complexity of the composed algorithm is
$f(|\cLM|)+  p(n) + O( n |\cLM|)$ as the time of computing the associated
instance is $O(n |\cLM|)$. Thus, the claim follows.
\end{proof}

\section{The computational complexity of trading $j$ coalitions}

Note that it is enough to check combinations of $j$-coalitions in  $\cL^M$ or  $\cW^m$ to seek for a  $j$-trade application.
This is so,
because, if needed,  we can remove players from a maximal losing coalition or  add  players
to a minimal wining coalition getting loosing or  winning coalitions that matches the requirements for the $j$-trade application.

\begin{theorem}\label{th:j-trade}
For a fixed $j$, given a simple game $\Gamma=(N,\cW)$, we can decide
whether such a simple game is $j$-trade in polynomial time. Furthermore,  if the game is $j$-trade,
a  $j$-trade application  can be \emph{efficiently computed}.
\end{theorem}

\begin{proof}
We start by computing $\cL^M$ from $\cW$ using the  polynomial time algorithm given in~\cite{FMOS11}.
Then, we try all  the possible combinations of $j$ members of $\cW$ and $j$ members of $\cL^M$ (repetitions are allowed)
to find a $j$-trade application. This requires  ${|\cW|}\choose{j}$ combinations.

The algorithm  works because there is such $j$-trade application if and only if there is a $j$-trade application where all losing coalitions
are maximal losing coalitions.
%Note that if we have a $j$-trade application and one of any losing coalition is not maximal we can add a player to one of any
%winning coalition and to one of any not maximal losing coalitions, and still have a valid $j$-trade application.
\end{proof}

Next we present a similar result but for the case in which  the game is given by $\Gamma=(N,\cL)$.

\begin{theorem}
For a fixed $j$, given a simple game $\Gamma=(N,\cL)$, we can decide
whether such a simple game is $j$-trade in polynomial time. Furthermore,  if the game is $j$-trade,
a  $j$-trade application  can be \emph{efficiently computed}.
\end{theorem}

\begin{proof}
We start by computing $\cW^m$ from $\cL$ in polynomial time~\cite{FMOS11} and proceed as Theorem~\ref{th:j-trade}.
\end{proof}

Now we adapt the result of Proposition~\ref{prop:2trade} for
the \OTR{\alpha}{\beta} problem to the \TR{\alpha}{\beta}{j}
problem
\begin{proposition}
\label{prop:jtrade}
The \TR{\alpha}{\beta}{j} problem is \NP-complete when $\alpha(\Gamma)=\cLM$
and $\beta(\Gamma)=\cL$.
\end{proposition}

\begin{proof}
The argument is quite similar to the one used in the  proof of Proposition~\ref{prop:2trade}, but
considering that the set of players of the game $\Gamma$ contains the literals and
$2+2\cdot(j-2)$ extra players, i.e.,  $a$, $b$, and $c_{1i}$ and $c_{2i}$ for
$i\in\{3,\ldots,j\}$. Thus, we have that $N = X \cup \{a, b\} \cup_{i=3}^j \{c_{1i},c_{2i}\}$. A set of
  players $Y$ can win if and only if at least one of the following
  conditions are met:
\begin{equation}
a \in Y \wedge \forall i = \{ 1, 2, \ldots, m \}: Y \cap X_i \neq \emptyset
\end{equation}
\begin{equation}
b \in Y \wedge \forall j = \{ 1, 2, \ldots, n \}: x_j \in Y \vee \neg x_j \in Y
\end{equation}
\begin{equation}
\label{eq:ci}
\{c_{1i}, c_{2i}\} \subseteq Y, \forall i\in\{3,4,\ldots,j\} \enspace .
\end{equation}

As winning coalitions we have
\[
\begin{array}{l}
 S_{j+1}=\{a, \textit{all true literals corresponding to}\\
 \qquad\qquad\textit{ a truth-assignment satisfying $\phi$}\},\\
 S_{j+2}=\{b, \textit{the remaining literals that are}\\
 \qquad\qquad\textit{ not in } S_1\},\\
 S_{j+i}=\{c_{1i}, c_{2i}\}\ ,\ \forall{}i\in\{3,4,\ldots,j\}.
\end{array}
\]
Finally,
as losing coalitions we distinguish two cases:
\begin{itemize}
\item{}If $j$ is even, %(par)
\[
\begin{array}{l}
 S_1=X,\\
 S_2=\{a, b\},\\
 S_i=\{c_{1i}, c_{1,{i+1}}\}\ ,\ \forall{}i\in\{3,5,7,\ldots,j-1\}\\
 S_i=\{c_{2,{i-1}}, c_{2i}\}\ ,\ \forall{}i\in\{4,6,8,\ldots,j\}\ .
\end{array}
\]
\item{}If $j$ is odd %(impar)
\[
\begin{array}{l}
 S_1=X,\\
 S_2=\{a, c_{1j}\},\\
 S_i=\{c_{1i}, c_{1,{i+1}}\}\ ,\ \forall{}i\in\{3,5,7,\ldots,j-1\}\\
 S_i=\{c_{2,{i-1}}, c_{2i}\}\ ,\ \forall{}i\in\{4,6,8,\ldots,j-2\}\\
 S_{2j}=\{b, c_{2j}\}\ .
\end{array}
\]
\end{itemize}
\end{proof}

\section{Related Remaining Problems}
In this paper we have focused on the computational complexity of trade robustness problems for  simple games.
Our results are summarized in Table~\ref{tab:sum}.
Nevertheless, there remain many related open questions.
Let us highlight some of them.

Because constructing  $\cW$ or $\cL$ from $\cW^m$ or $\cL^M$, respectively,
is not polynomially solvable,
we post the following conjecture.%, that so far we have been unable to prove.
\begin{conjecture}
For a fixed $j$, to decide whether a game given by either $(N, \cW^m)$ or $(N, \cL^M)$
is $j$-trade is \coNP-complete.
\end{conjecture}

%Next, we present two open problems.
%\comment{I cannot remmember what are the additional problems of interest but the two following
%seem very strange to me. Can you remember?}
%\begin{OP}
%\comment{Proposed by Maria.}
%Is this FPT or hard for some level of the W hierarchy?
%\end{OP}
We recall that a simple game is weighted if and only if it is
$j$-trade robust for any non-negative integer $j$, see the
characterization given by Taylor and Zwicker~\cite{TaZw99}. This leads us to the following problem

\begin{quote}
Trade robustness problem.

\texttt{Input:} A simple game $\Gamma$ and a non-negative integer $j$.

\texttt{Question:}
Is $\Gamma$ a $j$-trade robust game?
\end{quote}
Whose computational complexity remains open for the different  forms of representations of simple games considered in this paper.
%Note that there exists previous results related with this last problem when $j=2$.
%For instance,
%it is known that the problem of $2$-linear separable, which is equivalent to $2$-trade robustness for simple games when  both $\cW$ and $\cL$ are given,  $( is \NP-complete~\cite{Me88},
%See http://citeseerx.ist.psu.edu/viewdoc/summary?doi=10.1.1.31.2967
%and that the problem of $2$-monotonic positive boolean function is $O(n^2 m)$~\cite{MaIb98}.
%See http://www.sciencedirect.com/science/article/pii/S0196677497908968

In a recent paper Molinero \emph{et al.}~\cite{MRS15}  introduced \emph{influence games}. Influence games provide a succinct form of representation for  simple games based on graphs. It would be of interest to analyze the complexity of the $(\alpha,\beta,j)$-trade and the trade robustness problems when the simple game is given as an influence game.

%\section*{References}
\bibliographystyle{plain}
%\bibliographystyle{alpha}
%\bibliographystyle{abbrv}

%\bibliography{Biblio}

\end{document}